\documentclass{ws-ijbc}

\usepackage{hyperref}
\usepackage{booktabs}
\usepackage{float}
\usepackage{framed}
\usepackage[T1]{fontenc}
\usepackage{enumitem}
\usepackage{multicol}
\usepackage{ws-rotating}     
\usepackage{graphicx}
\usepackage{epstopdf}
\usepackage{tabularx}
\usepackage{colortbl}
\usepackage{arydshln}
\usepackage[toc,page]{appendix} 

\DeclareMathOperator{\diag}{diag}

\usepackage[titlenumbered,ruled]{algorithm2e}
\usepackage{listings}
\usepackage{lmodern}
\usepackage[usenames,dvipsnames]{xcolor}

\usepackage{algorithmicx}
\usepackage{algpseudocode}

\usepackage[toc,page]{appendix}

\algrenewcommand\alglinenumber[1]{\tiny #1:}
\usepackage{mathtools}
\usepackage{amsmath}%
\usepackage{amsfonts}%
\usepackage{amssymb}%
\usepackage{graphicx}
\usepackage{color}
\usepackage{cases}
\def\HiLi{\leavevmode\rlap{\hbox to \hsize{\color{yellow!50}\leaders\hrule height .8\baselineskip depth .5ex\hfill}}}

\newlist{myenumi}{description}{10}
\setlist[myenumi]{labelindent=\parindent, leftmargin=*, label=(\roman*), align=left}
\setlist[myenumi]{leftmargin=0pt}

\newsavebox{\listingbox}

\definecolor{darkgreen}{rgb}{0.1, 0.5, 0.1}
\begin{document}
\catchline{}{}{}{}{} 


\title{Improved Matlab code for Lyapunov exponents of fractional order systems}

\author{MARIUS-F. DANCA}

\address{STAR-UBB Institute,\\
Babes-Bolyai University, Cluj-Napoca, Romania
\\email: m.f.danca@gmail.com}

\maketitle

\begin{history}
\received{(to be inserted by publisher)}
\end{history}

\begin{abstract}
\begin{abstract}
This paper presents an improved Matlab routine, \texttt{FO\_LE}, for numerical computation of the Lyapunov exponents of fractional-order systems modeled by Caputo's derivative, conceived as an enhanced version of the former \texttt{FO\_Lyapunov} and \texttt{FO\_NC\_Lyapunov} codes for commensurate and non-commensurate orders, respectively \cite{dd2,dd3}. The proposed approach replaces the Gram--Schmidt orthogonalization procedure with QR-based reorthonormalization, and uses the new quadratic \texttt{LIL} predictor--corrector scheme for the integration of the extended variational system. Compared with the former implementations, the present routine benefits from the higher-order of the fractional integrator, \texttt{LIL}, and applies to both commensurate and non-commensurate models. As the previous code, \texttt{FO\_LE} retains the full memory structure of the underlying Caputo model. The Matlab code for the \texttt{LIL} solver and for the computation of Lyapunov exponents \texttt{FO\_LE} are provided, while a fast implementation of \texttt{LIL}, for commensurate and non-commensurate orders, \texttt{LIL\_nc}, is available on MathWorks File Exchange \cite{lilnc}. A benchmark problem with exact solution is used to compare the \texttt{LIL}-based solver against ABM-type methods, whereas the Rabinovich--Fabrikant system illustrates the computation of Lyapunov exponents in different dynamical regimes. The results indicate that the proposed implementation is a compact, robust, and efficient tool for the numerical study of stability and chaos in fractional-order systems.
\end{abstract}
\end{abstract}



\section{Introduction}
\begin{multicols}{2}

Despite its long history, fractional calculus was for many years only rarely used in mainstream physics and engineering, partly because fractional-order (FO) derivatives were often regarded as lacking a transparent geometrical or physical interpretation (see, e.g., \cite{pod}). In recent years, however, this situation has changed significantly. Fractional models are now increasingly recognized as effective tools for describing memory, hereditary effects, anomalous transport, and complex multiscale dynamics, and they have found applications in mechanics, electronics, control, finance, chemistry, and biophysics. For mathematical results on existence, uniqueness, continuous dependence on parameters, and asymptotic stability of solutions of fractional differential equations (FDEs) with general nonlinearities, see, for example, \cite{old,cap,existenta,kill,pod2,gore}.

Lyapunov exponents (LEs) characterize the average exponential rates at which nearby trajectories separate or approach each other in phase space. From a physical viewpoint, they provide quantitative information on predictability, stability of attractors and invariant sets, and the possible presence of sensitive dependence on initial conditions, one of the main signatures of chaos.

 At the same time, the physical significance and practical usefulness of LEs have also been debated. Their numerical approximation may also be affected by non-negligible computational errors, especially in systems with long memory and strong sensitivity to discretization. In this respect, Sprott noted that a significant number of LEs reported in the literature are either incorrect or stated with unjustified precision.

The computation of LEs remains a central topic in nonlinear dynamics and computational physics. In practice, LEs continue to be widely used as diagnostic indicators in the numerical investigation of dissipative, oscillatory, and chaotic regimes, and a variety of efficient algorithms have been developed for their approximation. By contrast, for non-commensurate FO systems, the literature is still limited and dedicated numerical codes remain scarce. This makes the development of robust and efficient algorithms for computing LEs in the non-commensurate FO framework particularly relevant.

\section{Problem formulation and variational system}

We consider the Initial Value Problem (IVP) for a system of Caputo fractional
differential equations
\begin{equation}
{}^{C}D_t^{\alpha_i}x_i(t)=f_i(t,x(t)),
\label{eq:gen_system_components}
\end{equation}
for \(i=1,2,\dots,n_e\), with initial condition
\begin{equation}
x(0)=x_0\in\mathbb{R}^{n_e},
\label{eq:gen_system_ic}
\end{equation}
where
\[
x(t)=\bigl(x_1(t),x_2(t),\dots,x_{n_e}(t)\bigr)^T
\]
and
\[
\alpha=(\alpha_1,\alpha_2,\dots,\alpha_{n_e})^T,
\qquad 0<\alpha_i<1.
\]

The above formulation includes both the non-commensurate and commensurate case, when
\[
\alpha_1=\alpha_2=\cdots=\alpha_{n_e}=\alpha\in(0,1),
\]
and the non-commensurate case, when at least two components of \(\alpha\)
are different.

For \(0<\alpha_i<1\), the Caputo derivative of \(x_i\) is defined by
\begin{equation}
{}^{C}D_t^{\alpha_i}x_i(t)
=
\frac{1}{\Gamma(1-\alpha_i)}
\int_0^t
\frac{x_i'(\tau)}{(t-\tau)^{\alpha_i}}\,d\tau,
\label{eq:caputo_definition_general}
\end{equation}
for \(i=1,\dots,n_e\). The Caputo framework is especially convenient in
applications because it allows the use of classical initial conditions of
the form \eqref{eq:gen_system_ic}. Throughout this work, the Caputo fractional derivative is understood with
initial point \(0\).
It is well known that
\eqref{eq:gen_system_components}-\eqref{eq:gen_system_ic} is equivalent to
the Volterra integral equation
\begin{equation*}
x_i(t)
=
x_i(0)
+
\frac{1}{\Gamma(\alpha_i)}
\int_0^t
(t-\tau)^{\alpha_i-1}f_i(\tau,x(\tau))\,d\tau,
\label{eq:volterra_general_components}
\end{equation*}
for \(i=1,\dots,n_e\). This representation is the basis of most
predictor--corrector schemes for fractional differential equations,
including the methods used in this paper.

To compute LEs, one needs in addition the variational system
associated with the IVP \eqref{eq:gen_system_components}-\eqref{eq:gen_system_ic}. Let \(x(t)\) be a
reference trajectory and let \(\Phi(t)=(\phi_{ij}(t))_{i,j=1}^{n_e}\) be
the perturbation matrix. Then, the linearized dynamics along \(x(t)\) is
described by
\begin{equation}
{}^{C}D_t^{\alpha_i}\phi_{ij}(t)
=
\sum_{k=1}^{n_e}
\frac{\partial f_i}{\partial x_k}(t,x(t))\,\phi_{kj}(t),
\label{eq:general_variational_components}
\end{equation}
for \(i,j=1,\dots,n_e\), together with the initial condition
\begin{equation}
\Phi(0)=I_{n_e},
\label{eq:general_variational_ic}
\end{equation}
where \(I_{n_e}\) denotes the identity matrix. Thus, the original system and
its variational equations define an extended system consisting of
\(n_e+n_e^2\) equations.

The idea is to integrate the extended system over
successive intervals of length \(h_{\mathrm{norm}}\), to orthonormalize the
perturbation vectors after each such interval, and to accumulate the
associated stretching factors. More precisely, if
\(t_{\mathrm{start}}=0\) and \(T\) denote the initial and final
times and if \(h_{\mathrm{norm}}>0\) is the renormalization step, then after
each integration step over \([t,t+h_{\mathrm{norm}}]\) one extracts the
updated perturbation matrix \(\Phi\), orthonormalizes its columns, and
obtains stretching factors \(z_k\), \(k=1,\dots,n_e\). These are used to
update the cumulative sums
\begin{equation*}
s_k \leftarrow s_k+\log z_k,
\qquad k=1,\dots,n_e.
\end{equation*}
The finite-time LEs are then approximated by
\begin{equation}
\lambda_k(t)=\frac{s_k}{t},
\qquad k=1,\dots,n_e.
\label{eq:le_general_formula}
\end{equation}
with $t\leftarrow t+h_{norm}$.

Therefore, the exponents are obtained as time averages of the logarithmic
stretching factors generated during the repeated renormalization process.

LEs which measure the exponential growth, or decay, of infinitesimal phase-space perturbations of a chaotic dynamical system can be obtained numerically by using the algorithm proposed in the seminal works of Benettin et al. \cite{bene2}, the first work where the Gram-Schmidt (GS) orthogonalization procedure is proposed to compute LEs for continuous systems of integer order (see also \cite{shima}, or \cite{baker} where the Benettin algorithm is presented in Basic langauge).
To avoid overflow, the algorithm calculates the divergence of nearby trajectories for finite timesteps and renormalizes to unity after a finite number of steps using the GS procedure \cite{lili,christ}).

The Benettin algorithm is presented in Algorithm \ref{alg1}.

\begin{remark}
The numerically obtained LEs of integer or fractional order, are generally finite-time and local, since they are computed over a finite interval and along the particular trajectory issued from the chosen initial condition, while the theoretical exponents are asymptotic quantities defined in the limit \(t\to\infty\).
\end{remark}

First Matlab codes to calculate LEs of the systems modeled by the IVP \eqref{eq:gen_system_components}-\eqref{eq:gen_system_ic} are presented in \cite {dd2} and \cite{dd3}, for commensurate order case and incommensurate order case, respectively (the codes can be downloaded at \cite{com}, \cite{noncom}). Both algorithms use the GS procedure.

In this framework, the commensurate and non-commensurate cases are treated
in exactly the same way at the algorithmic level. The only structural
difference lies in the order vector: in the commensurate case all equations
have the same order, whereas in the non-commensurate case the components of
\(\alpha\) may differ and the same multi-order structure has to be extended
to the variational equations. Apart from this, the Benettin
philosophy remains unchanged: one integrates the extended FO
system, orthonormalizes the perturbation basis, and updates the running
averages \eqref{eq:le_general_formula}.

In practical computations, the orthonormalization step can be performed
either by the classical GS procedure or by a QR decomposition of
the perturbation matrix as proposed in this paper.

To guarantee existence and uniqueness of solutions, we assume that
\(f:[0,T]\times\mathbb{R}^{n_e}\to\mathbb{R}^{n_e}\) is continuous and
satisfies a Lipschitz condition with respect to the state variable, namely,
there exists a constant \(L>0\) such that
\begin{equation*}
\|f(t,u)-f(t,v)\|\le L\|u-v\|,
\end{equation*}
for all \(t\in[0,T]\) and all \(u,v\in\mathbb{R}^{n_e}\). Under such
standard assumptions, the IVP admits a unique solution on
\([0,T]\) (see, for example, \cite{old,cap,existenta,kill,pod2}).

For the general non-commensurate case, it is natural to appeal to the
existence--uniqueness theory for nonlinear multi-order Caputo systems; see,
for example, \cite{diex}.

\begin{theorem}\label{ttt}
Assume that \(f:[0,T]\times\mathbb{R}^{n_e}\to\mathbb{R}^{n_e}\) is
continuous in \(t\), of class \(C^{1}\) with respect to \(x\), and that
both \(f\) and \(D_xf\) are locally Lipschitz in \(x\).
Let \(x(t)\)
be a solution of
\eqref{eq:gen_system_components}--\eqref{eq:gen_system_ic}. Then, the
associated variational system along \(x(t)\) is given by
\begin{equation}
\begin{split}
{}^{C}D_t^{\alpha_i}\phi_{ij}(t)
&=
\sum_{k=1}^{n_e}
\frac{\partial f_i}{\partial x_k}(t,x(t))\,\phi_{kj}(t),\\
&\qquad i,j=1,\dots,n_e,
\end{split}
\label{eq:general_variational_components_thm}
\end{equation}
with initial condition
\begin{equation}
\Phi(0)=I_{n_e},
\label{eq:general_variational_ic_thm}
\end{equation}
where \(\Phi(t)=(\phi_{ij}(t))_{i,j=1}^{n_e}\). Equivalently, the
right-hand side of \eqref{eq:general_variational_components_thm} is the
\((i,j)\)-entry of the matrix product
\[
D_xf(t,x(t))\,\Phi(t).
\]

Moreover, the coupled system consisting of
\eqref{eq:gen_system_components} and
\eqref{eq:general_variational_components_thm}, together with the initial
conditions \eqref{eq:gen_system_ic} and
\eqref{eq:general_variational_ic_thm}, defines an extended multi-order
Caputo system which admits a unique local solution on every interval on
which \(x(t)\) exists.
\end{theorem}

\begin{proof}
Since \(f\) is of class \(C^{1}\) with respect to \(x\), the Jacobian
matrix \(D_xf(t,x(t))\) exists and is continuous along the trajectory
\(x(t)\). Therefore, the linearization of
\eqref{eq:gen_system_components} along \(x(t)\) is well defined. In matrix
form, the variational dynamics is governed by the Jacobian
\(D_xf(t,x(t))\). More precisely, its right-hand side can be written as
\[
\bigl(D_xf(t,x(t))\,\Phi(t)\bigr)_{ij},
\]
and this yields, componentwise,
\eqref{eq:general_variational_components_thm}. The initial condition
\(\Phi(0)=I_{n_e}\) follows from the definition of the fundamental matrix.

Introduce now the extended unknown
\[
Y(t)=\bigl(x(t),\operatorname{vec}\Phi(t)\bigr)
\in\mathbb{R}^{n_e+n_e^2}.
\]
Then, the original system and the variational equations can be rewritten as
a nonlinear multi-order Caputo system of the form
\[
{}^{C}D_t^{\widetilde{\alpha}}Y(t)=F(t,Y(t)),
\]
where \(\widetilde{\alpha}\) is the extended order vector obtained by
repeating the orders \(\alpha_1,\dots,\alpha_{n_e}\) in the variational
block, and
\[
F(t,Y)=
\begin{pmatrix}
f(t,x)\\[1mm]
\operatorname{vec}\bigl(D_xf(t,x)\Phi\bigr)
\end{pmatrix}.
\]
By the assumptions on \(f\) and \(D_xf\), the mapping \(F\) is continuous
in \(t\) and locally Lipschitz in \(Y\) on bounded sets. Hence the
standard existence--uniqueness theorem for nonlinear multi-order Caputo
systems applies to the extended problem (see, for example, \cite{diex}).
Therefore, the extended system admits a unique local solution on every
interval on which the reference trajectory \(x(t)\) exists.
\end{proof}

The commensurate case of Theorem \ref{ttt} is
recovered by taking \(\alpha_1=\cdots=\alpha_{n_e}=q\).

\subsection{QR-based reformulation of the orthonormalization step}

In the classical Benettin algorithm, after each integration interval of length \(h_{\mathrm{norm}}\), the columns of the perturbation matrix are orthonormalized by the GS procedure, and the corresponding stretching factors are used to update the LEs.

In the present implementation, the overall Benettin framework is preserved, but the GS step is replaced by a QR decomposition. Such a replacement is natural, since for a perturbation matrix \(A\in\mathbb{R}^{n_e\times n_e}\), obtained after one renormalization interval, the reduced QR factorization
\[
A = QR,
\]
provides directly an orthonormal basis \(Q\) and an upper triangular matrix \(R\), whose diagonal entries contain the stretching information associated with the current step.

More precisely, let \(A\) denote the matrix whose columns are the evolved perturbation vectors after integrating the extended system on \([t,t+h_{\mathrm{norm}}]\). Instead of applying GS column by column, one computes
\[
[Q,R] = \mathrm{qr}(A,0),
\]
where the Matlab command $qr$ computes the reduced QR decomposition, \(Q\in\mathbb{R}^{n_e\times n_e}\) is orthonormal and \(R\in\mathbb{R}^{n_e\times n_e}\) is upper triangular. The orthonormalized perturbation basis is then given by the columns of \(Q\), which are used as initial perturbation directions for the next integration interval.

To ensure a consistent extraction of the stretching factors, the diagonal elements of \(R\) are taken positive. In practice, this is achieved by introducing the diagonal matrix
\[
D=\operatorname{diag}\big(\operatorname{sign}(\operatorname{diag}(R))\big),
\]
with the convention that zero entries are replaced by \(1\). Then, one redefines
\[
Q \leftarrow QD, \qquad R \leftarrow DR,
\]
so that the diagonal entries of the modified matrix \(R\) are nonnegative. The stretching factors used in the LE update are therefore
\[
z_k = |R_{kk}|, \qquad k=1,\dots,n_e.
\]
The cumulative sums are updated as
\[
s_k \leftarrow s_k + \log z_k,
\]
and the finite-time LEs are computed, as before, by
\[
\lambda_k(t)=\frac{s_k}{t}, \qquad k=1,\dots,n_e,
\]
where $t\leftarrow t+h_{norm}$.

Thus, the QR-based variant is fully equivalent in spirit to the usual GS renormalization procedure, the only difference being the numerical realization of the orthonormalization step. The advantage of the QR formulation is mainly practical: it leads to a more compact Matlab implementation and is generally more robust numerically than the explicit GS recursion.  If \(\Phi=QR\), where \(Q\) is orthonormal and \(R\) is upper triangular, then the columns of \(Q\) form the new perturbation basis, while the diagonal entries of \(R\) contain the stretching information used to update the LEs.

In practice, the QR decomposition is generally preferred to the classical
Gram--Schmidt procedure because it provides a more numerically stable and
robust orthonormalization, especially when the perturbation vectors become
nearly linearly dependent during long integrations.

The Algorithm \ref{alg1} presents the Benettin algorithm with GS procedure. To implement the QR decomposition, the yellow lines in Algorithm \ref{alg1} must be replaced by the following lines

{\setlength{\abovedisplayskip}{2pt}
\setlength{\belowdisplayskip}{2pt}
\setlength{\abovedisplayshortskip}{2pt}
\setlength{\belowdisplayshortskip}{2pt}
\begin{itemize}[leftmargin=1.5em,itemsep=0pt,topsep=0pt,parsep=0pt]
\item compute the reduced QR factorization
\[
\Phi = QR;
\]
\item enforce nonnegative diagonal entries in \(R\) and set
\[
z_k = |R_{kk}|,\qquad k=1,\dots,n_e;
\]
\item replace the perturbation basis by the columns of \(Q\).
\end{itemize}
}

\section{LIL Predictor-Corrector method}

In \cite{dd2} and \cite{dd3} the LEs for the commensurate and non-commensurate case, respectively, are determined via the Benettin algorithm with GS procedure and using the ABM method for FDEs \cite{kai}.

In the present work, the Benettin framework to find LEs is used, via:

\begin{itemize}[leftmargin=1.5em,itemsep=0pt,topsep=0pt,parsep=0pt]
\item the quadratic Lagrange Interpolation at the Last step (LIL) predictor--corrector scheme \cite{dd5};
\item  the QR factorization.
\end{itemize}
The LIL scheme, a high order method \emph{implicit predictor--corrector method
with full memory} (see Box 1), is a \emph{quadratic two-step backward interpolation} method for FDEs modeled by Caputo's derivative. The analysis and the properties of the scheme for the commensurate case, are  presented in \cite{dd5}. The Matlab code for the commensurate case is presented in Appendix \ref{a} and is called as

\vspace{4mm}
\noindent \fbox{\texttt{[t,x] = LIL(fun,tspan,x0,alpha,h)}}
\vspace{4mm}

where $fun$ is the function modelling the system, $tspan=[0,T]$, $\alpha$ is the fractional order, $x0$ the initial condition, and $h$ the integration step size. Optional input and output parameters are described in the Appendix \ref{a}.

A faster variant \texttt{LIL} for the commensurate case can be found at Mathworks \cite{lil}, while for the non-commensurate case, a faster code, \texttt{LIL\_nc}, can be found at Mathworkd \cite{lilnc}.

A typical command to run the \texttt{LIL\_nc} is

\vspace{4mm}
\noindent \fbox{\texttt{[t,x] = LIL\_nc(fun,tspan,x0,alpha,h)}}

\vspace{4mm}
where $\alpha=(\alpha_1,\alpha_2,...,\alpha_{ne})$, with $ne$ the size of the system. Here also optional parameters can be used.

\vspace{3mm}
\emph{The non-commensurate variant, \emph{\texttt{LIL\_nc}} utilized in this paper, can be used for both commensurate and non-commensurate cases and.
}
\vspace{3mm}

\begin{remark}
As shown in \cite{dd5} under suitable smoothness assumptions, the LIL method may reach a third-order accuracy and also may provide a more accurate approximation than the Adams--Bashforth--Moulton (ABM) scheme for the same step size. It should also be emphasized that the Matlab solvers \texttt{FDE12} and \texttt{fde12\_nc2}\footnote{The code \texttt{fde12\_nc2} is a non-commensurate predictor--corrector implementation belonging to the \texttt{fde12} family of FFT-accelerated product-integration PECE solvers introduced and documented by Garrappa; see, in particular, the survey on numerical solution of fractional differential equations in \cite{gara}).} used in the previous Matlab codes for LEs \cite{com} and \cite{noncom}, are not basic ABM implementations, but fast FFT-based variants of the ABM predictor--corrector method.
\end{remark}

Consider the following benchmark test of the schemes \texttt{LIL\_nc}, \texttt{fde12\_nc2} and \texttt{ABM\_nc} (non-commensurate variant of the genuine ABM method) to assess the performances in the case of the following Caputo two-dimensional non-commensurate fractional system for $t\in[0,1]$:

\begin{equation}
\label{eq:R2_nc_test}
\left\{
\begin{aligned}
{}^C D_t^{\alpha_1} x_1(t) &= -x_1(t),\\[1mm]
{}^C D_t^{\alpha_2} x_2(t) &= -2x_2(t)
+\bigl(x_1(t)-x_{1,e}(t)\bigr)^2
\end{aligned}
\right.
\end{equation}
with initial condition
\begin{equation}
x_1(0)=x_2(0)=1,
\label{eq:R2_nc_ic}
\end{equation}
and non-commensurate orders
\begin{equation*}
\alpha_1=0.9,\qquad \alpha_2=0.8.
\label{eq:R2_nc_alpha}
\end{equation*}

Here
\begin{equation*}
x_{1,e}(t)=E_{\alpha_1}(-t^{\alpha_1}),
\label{eq:R2_nc_x1e}
\end{equation*}
where $E_\alpha(\cdot)$ denotes the one-parameter Mittag--Leffler function
\[
E_\alpha(z)=\sum_{k=0}^{\infty}\frac{z^k}{\Gamma(\alpha k+1)}.
\]

The first equation in \eqref{eq:R2_nc_test} is the standard linear Caputo
test equation, whose exact solution is
\[
x_1(t)=E_{\alpha_1}(-t^{\alpha_1}).
\]
Consequently, along the exact solution the nonlinear correction term in the
second equation vanishes identically, and the second component reduces to
\[
{}^C D_t^{\alpha_2}x_2(t)=-2x_2(t),\qquad x_2(0)=1,
\]
with exact solution
\[
x_2(t)=E_{\alpha_2}(-2t^{\alpha_2}).
\]
Therefore, the exact solution of \eqref{eq:R2_nc_test}--\eqref{eq:R2_nc_ic}
is (Fig. \ref{fig1})
\begin{equation*}
x_e(t)=
\begin{bmatrix}
E_{\alpha_1}(-t^{\alpha_1})\\[1mm]
E_{\alpha_2}(-2t^{\alpha_2})
\end{bmatrix}.
\end{equation*}

This benchmark is convenient for validating the code for several reasons.
First, it is genuinely non-commensurate. Second, it is mildly nonlinear, due to the
quadratic correction in the second equation. Third, despite this
nonlinearity, the exact solution is known explicitly in terms of
Mittag--Leffler functions, so that the numerical error can be evaluated
directly, fact which allows a direct comparison.

Table~\ref{tab:R2_nc_LIL_ABM} is obtained through a standard mesh-refinement benchmark for the considered non-commensurate test problem. More precisely, the problem is solved with the three numerical methods \texttt{LIL\_nc}, \texttt{fde12\_nc2}, and classical \texttt{ABM\_nc} on a sequence of successively refined uniform meshes, corresponding to the step sizes
\[
h=0.01,\ 0.005,\ 0.0025,\ 0.00125,\ 0.000625.
\]
For each value of \(h\), the numerical solution is computed on the grid \(t_n=t_0+nh\), and then compared with the exact solution evaluated at the same grid points. The reported errors are measured in the discrete maximum norm, namely
\[
E(h)=\max_n \|x_n-x(t_n)\|_\infty,
\]
where \(x_n\) denotes the numerical approximation and \(x(t_n)\) the exact solution at the node \(t_n\). In this way one obtains, for each method, the error columns \(E_{\mathrm{LIL}}\), \(E_{\mathrm{fde12\_nc2}}\), and \(E_{\mathrm{ABM}}\).

The columns \(p_{\mathrm{LIL}}\), \(p_{\mathrm{fde12\_nc2}}\), and \(p_{\mathrm{ABM}}\) represent the experimental convergence orders, computed whenever two consecutive mesh levels are available. These values are determined by the standard formula
\[
p=\frac{\log\!\big(E(h)/E(h/2)\big)}{\log 2}
=\log_2\!\left(\frac{E(h)}{E(h/2)}\right).
\]
Therefore, the first row contains no value for \(p\), since no previous refinement level is available for comparison.

Finally, the CPU columns are obtained by measuring the execution time of each solver for every step size. Hence, each row of the table summarizes, for a fixed step size \(h\), the corresponding error, the observed convergence order, and the computational time for the three methods under comparison.

The call of \texttt{LIL\_nc} for the numerical solution of the system \eqref{eq:R2_nc_test}-\eqref{eq:R2_nc_ic} is

\begin{lstlisting}
[t,x] = LIL_nc(@test_R2,[0,1],[1,1],[0.9,0.8],...
0.001);
\end{lstlisting}
where the function \texttt{test\_R2} is
\begin{lstlisting}
function dx = test_R2(t,x)
alpha1 = 0.9;
x1e = mlf(alpha1,1,-t^alpha1);
dx = zeros(2,1);
dx(1) = -x(1);
dx(2) = -2*x(2) + (x(1)-x1e)^2;
end
\end{lstlisting}

The parameter \(\alpha_2\) does not appear explicitly in the function
 \texttt{test\_R2} because this function defines only the right-hand side
\(f(t,x)\) of the system, whereas the fractional orders are handled separately by the numerical solver. In particular, the second component is given by
\[
f_2(t,x)=-2x_2+(x_1-x_{1,e}(t))^2,
\]
which does not involve \(\alpha_2\) explicitly. By contrast, \(\alpha_1\) is needed in the code because the exact solution
\[
x_{1,e}(t)=E_{\alpha_1}(-t^{\alpha_1})
\]
appears in the forcing term. Thus, \(\alpha_2\) affects the dynamics through the operator
\({}^C D_t^{\alpha_2}\), but not through the explicit expression of the vector field.

In the comparison of \texttt{LIL\_nc} with \texttt{fde12\_nc2}, the latter
yields the smallest errors for all tested step sizes, however, the difference with respect
to \texttt{LIL\_nc} remains moderate: both solvers produce errors of the same
order of magnitude, and for the finest discretizations the gap becomes quite
small. At the same time, \texttt{LIL\_nc} exhibits a very stable observed
convergence rate, close to $1.61$ throughout the refinement process, whereas
the experimental orders of \texttt{fde12\_nc2} remains under
$1.57$, a not negligible characteristic. Indeed, as shown in \cite{dd5}
the LIL approach is designed to possess, under suitable
smoothness assumptions, a higher convergence order than the classical
Adams--Bashforth--Moulton predictor--corrector method. Since
\texttt{fde12\_nc2} is itself an ABM-type code, the fact that
\texttt{LIL\_nc} retains here a slightly larger observed order is consistent
with that theoretical advantage.

Moreover, the comparison is particularly meaningful because
\texttt{fde12\_nc2} is not a naive ABM implementation: it is an optimized
solver in which the history terms are evaluated by FFT-based fast
convolution. Hence, the CPU times reported for \texttt{fde12\_nc2}
correspond to an already accelerated predictor--corrector realization. In
this sense, the results show that \texttt{LIL\_nc} remains competitive with
an optimized ABM-based code: although slightly less accurate in absolute
error, it preserves a somewhat better numerical convergence rate while
requiring a comparable, and in some cases smaller, computational cost.

Therefore, Table~\ref{tab:R2_nc_LIL_ABM} highlights two complementary facts.
First, the proposed \texttt{LIL\_nc} solver is clearly superior to the
classical non-accelerated \texttt{ABM\_nc} implementation, both in accuracy
and in CPU time. Second, when compared with the optimized FFT-based solver
\texttt{fde12\_nc2}, \texttt{LIL\_nc} remains highly competitive: although
\texttt{fde12\_nc2} is slightly more accurate on this benchmark,
\texttt{LIL\_nc} achieves errors of the same order of magnitude, with a
comparable computational cost and a slightly higher observed convergence
order. This confirms that the proposed LIL-based approach is an effective
alternative for non-commensurate fractional systems.

\section{Matlab Routine FO\_LE for LEs}
The Matlab routine \texttt{FO\_LE} designed to compute numerical finite-time LEs
of Benettin type for the FO system modeled by the IVP \eqref{eq:gen_system_components}-\eqref{eq:gen_system_ic} and the variational equations \eqref{eq:general_variational_components}-\eqref{eq:general_variational_ic} uses the QR decomposition. More precisely, the
original system and its variational equations are integrated simultaneously,
and after each renormalization interval \(h_{\mathrm{norm}}\) a QR
orthonormalization is performed. The exponents are then estimated from the
time-averaged logarithmic growth factors associated with the diagonal entries
of the \(R\) matrix. Therefore, the reported values should be interpreted as
finite-time numerical approximations of the asymptotic Lyapunov spectrum of
the fractional variational dynamics.

The method combines the Benettin algorithm with the solver
\texttt{LIL\_nc} \cite{lilnc} used to integrate the extended system.

More precisely, for each renormalization interval determined by
\(h_{\mathrm{norm}}\), the extended system is integrated by the \texttt{LIL\_nc} scheme. At the end of each interval, the
variational block is reshaped into a matrix \(A\in\mathbb{R}^{n_e\times
n_e}\), whose columns are the evolved perturbation vectors.
The routine then performs a reduced QR decomposition
\[
A=QR.
\]
The orthonormal factor \(Q\) provides the new perturbation basis, while the
diagonal entries of the upper triangular factor \(R\) yield the local
stretching rates. The LEs are obtained by accumulating the
quantities \(\log|R_{kk}|\) and dividing by the elapsed time. The matrix
\(Q\) is then reinserted into the extended state, and the procedure is
repeated until the final time is reached.

Hence, the algorithm consists of an alternation between fractional
integration by \texttt{LIL\_nc} and QR-based reorthonormalization, which
allows stable and efficient numerical approximation of the full Lyapunov
spectrum.

\begin{remark}
The routine \texttt{FO\_LE} can be interpreted as an impulsive algorithm
\cite{md1}. More precisely, the numerical procedure alternates between the
continuous-time integration of the extended FO system and
discrete orthonormalization actions performed successively on the
subintervals \([kh_{\mathrm{norm}},(k+1)h_{\mathrm{norm}}]\), \(k=0,1,\dots\),
where \(h_{\mathrm{norm}}\) is taken as a multiple of the integration step
size \(h\). At the end of each subinterval, the last computed state \(x\) is
used as the initial value for the next integration stage, after the
variational block has been modified by the QR decomposition. Although each
new integration stage starts at the current impulsive time, the Lyapunov
exponents must still be normalized with respect to the global initial time.
Therefore, \texttt{FO\_LE} naturally fits into an impulsive formulation.

As proved in \cite{md1}, this type of algorithm preserves the full memory
principle. Hence, similarly to the classical Gram--Schmidt-based approach,
the QR reorthonormalization steps do not destroy the fractional memory
structure, but are consistent with a changing-lower-limit impulsive
framework.
\end{remark}

To call the \texttt{FO\_LE} routine, the following command could be used
\begin{lstlisting}
[t,LE] = FO_LE(ne,ext_fcn,t_start, h_norm,...
t_end,x_start,h,out);
\end{lstlisting}
where $ne$ represents the number of equations modelling the system, $out$ is the printing step of LE values (optional, default $0$). Other supplementary, optional parameters, can be used (Appendix \ref{b})

\subsection{Examples}

Consider the FO RF system \cite{dd4}
\begin{equation}
\begin{cases}
{}^{C}D_t^{\alpha_1}x_1 = x_2(x_3 - 1 + x_1^2)+a x_1,\\[2mm]
{}^{C}D_t^{\alpha_2}x_2 = x_1(3x_3 + 1 - x_1^2) +a x_2,\\[2mm]
{}^{C}D_t^{\alpha_3}x_3 = -2x_3(b + x_1x_2),
\end{cases}
\label{eq:system}
\end{equation}
which for $a=-1$ and $b=-0.1$ has only 3 equilibria, compared to the existing 5 equilibria for other parameter values.

The equilibria of the system \eqref{eq:system} are
\[
E_0=(0,0,0),
\]
and
\[
E_{\pm}=
\left(
\pm 0.1479,\,
\pm 0.6759,\,1.1969
\right).
\]

\begin{proposition}
The trivial equilibrium \(E_0\) is unstable for any
fractional orders \(\alpha_1,\alpha_2,\alpha_3\in(0,1)\).
\begin{proof}
The Jacobian at \(E_0\) is
\[
J(E_0)=
\begin{pmatrix}
-1 & -1 & 0\\
1 & -1 & 0\\
0 & 0 & 0.2
\end{pmatrix},
\]
hence the linearized third equation decouples as
\[
{}^C D_t^{\alpha_3}x_3=0.2\,x_3.
\]
Since \(0.2>0\), this scalar fractional equation is unstable for every
\(\alpha_3\in(0,1)\).
The Mittag-Leffler function $E_\alpha(\lambda t^\alpha)$ grows faster than any polynomial for large $t$ when $\lambda>0$. Therefore, \(E_0\) is always unstable.
\end{proof}
\end{proposition}

\begin{remark}\label{rama}
The computed Lyapunov spectrum should be interpreted as the spectrum of the
numerically generated orbit starting from the prescribed initial condition.
Consequently, when the system admits several possible long-time behaviors, one
cannot determine beforehand which equilibrium, oscillatory regime, or chaotic
set this spectrum is associated with; this becomes clear only a posteriori,
from the observed asymptotic evolution of the trajectory.
\end{remark}

As known, for a commensurate fractional-order system, the stability criterion of an equilibrium $E$ states that it is locally asymptotically stable if for every eigenvalues $\lambda$
\[
|\arg(\lambda)|>\alpha\frac{\pi}{2}.
\]
If for some $\alpha$
\[
|\arg(\lambda)|<\alpha\frac{\pi}{2},
\]
the equilibrium is unstable, and if for some eigenvalue
\[
|\arg(\lambda)|=\alpha\frac{\pi}{2},
\]
the equilibrium may be marginally stable and requires further analysis.

In the incommensurate case $(\alpha_1,\alpha_2,...,\alpha_{ne})$, let $M$ the least common multiple (LCM) of the denominators of $\alpha_i$, $i=1,2,...,ne$, expressed as rational numbers and define
\[
\gamma=\frac{1}{M}.
\]
The equilibrium $E$ is locally asymptotically stable if
\begin{equation}\label{ee}
|\arg(\lambda)|>\gamma\frac{\pi}{2}
\end{equation}
for all roots $\lambda$ of the pseudo-polynomial equation
\begin{equation}\label{frac}
\det(\diag(\lambda^{M\alpha_1}, \lambda^{M\alpha_2},...,\lambda^{M\alpha_{ne}})-J(E))=0,
\end{equation}
where $J$ is the Jacobian.

\begin{enumerate}[leftmargin=0cm,labelsep=0cm,align=left]
\item[1) ]
Consider the commensurate order $\alpha_1=\alpha_2=\alpha_3=0.999$ (chosen close to 1 so that the dynamics of the system are close to the integer case).

In this case, the critical value is
\[
\theta^*=\alpha\frac{\pi}{2}\approx89.91^\circ
\]

The eigenvalues of $E_\pm$ are $\lambda_1\approx-2.15$, $\lambda_{2,3}=0.175\pm1.056 i$ and because $\arg(\lambda_{2,3})\approx80.6^\circ<\theta^*$, $E_\pm$ are unstable (saddle-focus points, due to the positive value of $\lambda_1$ and the complex values of $\lambda_{2,3}$, respectively).

The trajectory starting from the initial condition $(0.1,0.1,0.1)$ is chaotic (Fig. \ref{fig2} (a)), being obtained with the \texttt{LIL} command
\begin{lstlisting}
[t,x] = LIL(@(t,x)RF(t,x,-1,-0.1),[0,1000],...
[0.1,0.1,0.1],0.999,0.01);
\end{lstlisting}
The function \texttt{RF\_ext} is presented in Appendix \ref{c}.

The same result can be obtained with \texttt{LIL\_nc} for $\alpha=[0.999;0.999;0.999]$.

The LEs are obtained with the command
\begin{lstlisting}
[t,LE] = FO_LE(3,@RF_ext,0,0.2,1500,[...
    0.1;0.1;0.1],0.01,[0.999;0.999;0.999],500);
\end{lstlisting}
that gives the following result:
\[
LE=(0.1017,0.0000,-1.9048),
\]
The values, printed every $500\times h_{\mathrm{norm}}=100$, are given in the following table, and the time evolution of the LEs is presented in Fig. \ref{fig2} (b).
\begin{lstlisting}
      t           LE1           LE2           LE3
  100.0000   0.07723959  -0.03851780  -1.84169357
  200.0000   0.11416719  -0.08493444  -1.83229826
  300.0000   0.08996528  -0.00689660  -1.88609942
  400.0000   0.08766504  -0.00368054  -1.88701509
  500.0000   0.10138621  -0.00543670  -1.89897681
  600.0000   0.10195209  -0.00433533  -1.90063980
  700.0000   0.10120368   0.00135328  -1.90557693
  800.0000   0.10049402  -0.01295196  -1.89058138
  900.0000   0.09679697   0.00262462  -1.90243042
 1000.0000   0.09598765  -0.00109984  -1.89788413
 1100.0000   0.09253276  -0.00381408  -1.89172487
 1200.0000   0.09591835   0.00132116  -1.90024112
 1300.0000   0.09731926  -0.00005847  -1.90027265
 1400.0000   0.09787358  -0.00058637  -1.90029196
 1500.0000   0.10169812   0.00006493  -1.90477474
\end{lstlisting}

The positiveness of the largest LE, underlines the chaotic dynamics of the trajectory.

\vspace{3mm}
\item[2) ]
Consider $\alpha=[0.6;0.8;0.7]$

Since
\[
\alpha_1=\frac35,\qquad \alpha_2=\frac45,\qquad \alpha_3=\frac7{10},
\]
we obtain
\[
M=\operatorname{LCM}(5,5,10)=10,
\qquad
\gamma=\frac1M=\frac1{10}.
\]
Thus, the stability condition in the \(\lambda\)-plane is
\[
|\arg(\lambda)|>\theta^*
=\frac{\pi}{20}=9^\circ.
\]
The characteristic equation is polynomial of 21 degree. The smallest argument of the 21 roots $\lambda_k$, $k=1,2,...,21$ is
\[
\min_k|\arg(\lambda_k)|=12.223^\circ>\theta^*.
\]
Therefore, $E_{1,2}$ are locally asymptotically stable as shown by the trajectory in Fig \ref{fig1} (c), which tends to $E_{+}$, obtained with the command

\begin{lstlisting}
[t,x] = LIL_nc(@(t,x)RF(t,x,-1,-0.1),...
[0,500],[0.1,0.1,0.1],[0.6,0.8,0.7],0.01);
\end{lstlisting}

The LEs are given by
\begin{lstlisting}
[t,LE] = FO_LE(3,@RF_ext,0,0.2,1000,...
    [0.1;0.1;0.1],0.01,[0.6;0.8;0.7]);
\end{lstlisting}
\[
LE =(-0.0894,-0.1025,-2.9471),
\]
with the graph of the time evolution in Fig \ref{fig1} (d), in concordance with the graphical result.

\item[3) ] Consider $\alpha=(0.85,0.965,0.999)$.

Then
\[
M=1000
\]
and
\[
\gamma=\frac{1}{M}= \frac{1}{1000}
\]
The generalized
characteristic equation \eqref{frac} associated with the linearization at \(E_1\) and
\(E_2\) is
\[
\det\!\Bigl(
\operatorname{diag}(\lambda^{850},\lambda^{965},\lambda^{999})
-
J(E_{1,2})
\Bigr)=0,
\]
which expands to a 2814 degree polynomial equation. 
According to the generalized Matignon criterion, the equilibrium is
asymptotically stable if every root \(\lambda\) of this characteristic
equation satisfies \eqref{ee}
\[
|\arg(\lambda)|>\theta^*=\gamma\frac{\pi}{2}=\frac{\pi}{2000}\approx 0.0016\ \text{rad}\approx 0.09^\circ,
\]
i.e. the unstable sector is an extremely narrow wedge around the positive real
axis.

Numerically, the minim argument of all roots is

\[
\min_j|\arg(\lambda_j)|=0.0920>\theta^*,\quad j=1,2,...,2814.
\]
Therefore, all roots lie outside the unstable sector and $E_{-/+}$ satisfy the generalized stability
condition and are locally asymptotically stable.

The trajectory, after removed transients, is plotted in Fig. \ref{fig1} (d) and is obtained with the command
\begin{lstlisting}
[t,x] = LIL_nc(@(t,x)RF(t,x,-1,-0.1),...
[0,1000],[0.1,0.1,0.1],[0.85,0.965,0.999],0.01).
\end{lstlisting}

\begin{remark}
This trajectory does not represents a periodic stable trajectory, but an ``apparently'' (``numerically'') periodic, since non-constant periodic solutions do not exist in autonomous fractional systems modeled by Caputo's derivative \cite{tava}
\end{remark}

The LEs are obtained with the command

\begin{lstlisting}
[t,LE] = FO_LE(3,@RF_ext,0,1,1000,...
[-0.0831,0.1298,0.6658],0.01,[0.85,0.965,0.999]);
\end{lstlisting}
\[
LE =(-0.0007,-0.1303,-1.4903)
\]
with the time evolution presented in Fig. \ref{fig1} (e).
\end{enumerate}

Due to the high nonlinearity of the equations modelling the system, the LEs corresponding to the apparent periodic trajectory, can be determined only by considering the initial condition close to the trajectory \texttt{[-0.0831,0.1298,0.6658]} (see Remark \ref{rama}).
\subsection{Selection of the Main Numerical Parameters in the Routine FO\_LE}

The choice of the renormalization interval \(h_{\mathrm{norm}}\) remains, at
present, essentially empirical. Unlike the integration step \(h\), which is
directly related to the accuracy of the numerical solver \texttt{LIL\_nc}, the parameter
\(h_{\mathrm{norm}}\) determines the frequency of reorthonormalization in the
Benettin procedure, and its influence on the computed Lyapunov exponents is
more subtle. Although the QR decomposition itself does not provide a direct
theoretical formula for selecting \(h_{\mathrm{norm}}\), the two are closely
connected at the algorithmic level: \(h_{\mathrm{norm}}\) fixes the time span
over which the variational equations are integrated before each QR
reorthonormalization step. If \(h_{\mathrm{norm}}\) is too large, the
propagated perturbation vectors may become nearly linearly dependent, which
can deteriorate the numerical conditioning of the QR step. On the other hand,
if \(h_{\mathrm{norm}}\) is too small, reorthonormalization is performed too
frequently, and the associated finite-time stretching factors may become less
informative. Consequently, no universal theoretical criterion is currently
available for the optimal choice of \(h_{\mathrm{norm}}\), and its selection
must be based on numerical robustness considerations.

Our numerical experiments showed that no single rule
provides an optimal choice of $h_{norm}$ for all fractional-order configurations. In
particular, the most suitable value of \(h_{{norm}}\) depends on the
order vector \(\alpha\), so that a value performing well for one set of
orders may become less effective for another. For this reason, the choice of
\(h_{{norm}}\) in the present work was guided by a systematic
robustness study based on several complementary criteria: refinement with
respect to the integration step \(h\), variation of the renormalization step
\(h_{{norm}}\), extension of the final integration time \(T\), and
monitoring of the empirical robustness indicator \(E_{\max}\).

By
combining all the above criteria, we found that for the discretization
\(h=0.01\), values of \(h_{{norm}}\) around \(0.2\), that is,
\(h_{{norm}}\) taken as a multiple of the integration step \(h\),
provide a satisfactory compromise between stability, consistency of the
computed exponents, and computational efficiency.

The tests indicate that, within the ranges considered here, the computed
Lyapunov spectrum is substantially more sensitive to the choice of
\(h_{\mathrm{norm}}\) than to moderate variations of either \(h\) or the
final time \(T\). The final time \(T\) should be large enough for
the Lyapunov-exponent curves to enter a stable asymptotic regime and become
as smooth as possible, but not so large that accumulated numerical errors
begin to affect the results. In the present computations, the interval
\(T\in[1000,1500]\) was found to provide a reasonable compromise between
asymptotic stabilization and long-time numerical reliability \cite{sara,sara2}.

\section*{Conclusion}

This paper presents a Matlab routine for the numerical computation of
finite-time Lyapunov exponents of fractional-order systems in both the
commensurate and non-commensurate setting. The proposed implementation
preserves the classical Benettin framework, but replaces the usual
Gram--Schmidt orthonormalization by a QR-based procedure and uses the new memory-safe implicit quadratic LIL predictor--corrector solver
for the numerical integration of the extended variational system.

The resulting approach combines three ingredients: the variational
formulation for multi-order Caputo systems, repeated reorthonormalization of
the perturbation basis, and a full-memory implicit solver adapted to
fractional dynamics. In this way, the method remains faithful to the
theoretical structure of Lyapunov-exponent computation for systems with
memory, while providing a compact and numerically robust Matlab
implementation.

Moreover, the proposed method is an impulsive QR-based implementation which preserves the full memory principle of the fractional-order system, so that the repeated reorthonormalization steps do not destroy the intrinsic memory structure of the dynamics.

The benchmark test on a non-commensurate fractional system with known exact
solution shows that the proposed \texttt{LIL\_nc}-based approach is clearly
more efficient than a classical non-accelerated ABM implementation and
remains competitive with FFT-accelerated ABM-type solvers.

The examples based on the Rabinovich--Fabrikant system illustrate that the
routine can successfully detect different dynamical regimes, including
chaotic behavior and convergence toward asymptotically stable equilibria, in
both commensurate and non-commensurate cases. At the same time, the
computed spectrum must always be interpreted as the spectrum of the
numerically generated orbit corresponding to the chosen initial condition,
especially when several coexisting long-time regimes are possible.

Overall, the proposed code provides a new numerical alternative to the
previous ABM-based implementations for Lyapunov exponents of
fractional-order systems. It extends the available computational tools for
multi-order Caputo models and may serve as a useful basis for further
studies on stability, bifurcation, hidden attractors, and chaotic dynamics
in fractional-order systems.

\newcounter{appsec}
\renewcommand{\theappsec}{\Alph{appsec}}

\newcommand{\appsection}[2][]{%
  \refstepcounter{appsec}%
  \section*{Appendix \theappsec. #2}%
  \addcontentsline{toc}{section}{Appendix \theappsec. #2}%
  \ifx\\#1\\\else\label{#1}\fi
}

\newpage



\vspace{1.5em}
\appsection{LIL code}\label{a}
\begin{lstlisting}
function [t,y,itcount]=LIL(fun,tspan,y0,alpha,h,...
    tol,maxit)
%
% Implicit quadratic PECE solver for Caputo
% commensurate FDEs
%
%   D_t^alpha y(t) = f(t,y(t)),   0 < alpha < 1,
%   y(t0) = y0
%
% Marius-F. Danca, 2026
%
% Call:
%   [t,y,itcount] = LIL(fun,tspan,y0,alpha,h,tol,
%   maxit)
% INPUT
%   fun   - function handle, fun(t,y)
%   tspan - [t0,T]
%   y0    - initial condition (scalar or column
%           vector)
%   alpha - fractional order, 0<alpha<1
%   h     - time step
%   tol   - stopping tolerance for fixed-point
%           iteration (optional)
%   maxit - maximum allowed iterations per step
%           (optional)
% OUTPUT
%   t       - time grid
%   y       - numerical solution
%   itcount - number of fixed-point iterations at
%           each step
%   example:
%   [t,Y] = LIL(@(t,Y)lorenz_system(t,Y,10,28,8/3),
%           [0,50],[1,1,1],0.99,.01);
%   function lorenz_system must be in the same
%   folder with LIL.m.
%
%   cite as:
%
%   Marius-F. Danca, A High Order Method for Caputo
%   Fractional Differential Equations, manuscript,
%   2026
%
%   For non-commensurate fractional order, LIL_nc,
%   a fast routine, can be found at:
%
%   https://www.mathworks.com/matlabcentral/fileexchange/183478-lil_nc


    if nargin < 6 || isempty(tol)
        tol = 1e-12;
    end
    if nargin < 7 || isempty(maxit)
        maxit = 100;
    end
    if numel(tspan) ~= 2
        error('tspan must be [t0,T].');
    end
    t0 = tspan(1);
    T  = tspan(2);
    if ~(isscalar(T) && isscalar(t0) && T > t0)
        error(['tspan must satisfy tspan(2) > ' ...
            'tspan(1).']);
    end
    if ~(isscalar(alpha) && alpha > 0 && alpha < 1)
        error(['alpha must be a scalar with 0 < ' ...
            'alpha < 1.']);
    end
    if ~(isscalar(h) && h > 0)
        error('h must be a positive scalar.');
    end
    N = round((T - t0)/h);
    if N < 1
        error(['The step size h is too large for ' ...
            'the interval [t0,T].']);
    end
    if abs(t0 + N*h - T) > 1e-14 * max(1,abs(T))
        warning(['h does not divide [t0,T] exactly; ' ...
            'using N = round((T-t0)/h).']);
    end
    t = t0 + (0:N)'*h;
    y0 = y0(:);
    d  = length(y0);
    y = zeros(d,N+1);
    y(:,1) = y0;
    f0 = fun(t(1),y0);
    f0 = f0(:);
    if length(f0) ~= d
        error(['fun(t0,y0) must have the same ' ...
            'dimension ' 'as y0.']);
    end
    fval = zeros(d,N+1);
    fval(:,1) = f0;
    itcount = zeros(N,1);
    facP = h^alpha / gamma(alpha+1);
    facC = h^alpha / gamma(alpha);
    %-----------------------------------------
    % Predictor weights
    %-----------------------------------------
    b = zeros(N,1);
    for j = 0:N-1
        b(j+1) = (j+1)^alpha - j^alpha;
    end
    %-----------------------------------------
    % Exact algebraic interval contributions
    %-----------------------------------------
    J0all  = zeros(N,1);
    J1all  = zeros(N,1);
    Iplus  = zeros(N,1);
    Izero  = zeros(N,1);
    Iminus = zeros(N,1);
    for j = 0:N-1
        jp1 = j + 1;
        d0 = jp1^alpha     - j^alpha;
        d1 = jp1^(alpha+1) - j^(alpha+1);
        d2 = jp1^(alpha+2) - j^(alpha+2);
        A0 = d0 / alpha;
        A1 = jp1 * d0 / alpha - d1 / (alpha + 1);
        A2 = jp1^2 * d0 / alpha ...
           - 2 * jp1 * d1 / (alpha + 1) ...
           + d2 / (alpha + 2);
        J0all(j+1) = A0 - A1;
        J1all(j+1) = A1;
        Iplus(j+1)  = 0.5 * (A2 + A1);
        Izero(j+1)  = A0 - A2;
        Iminus(j+1) = 0.5 * (A2 - A1);
    end
    %--------------------------------------------
    % Time stepping
    %--------------------------------------------
    for n = 0:N-1
        % Predictor
        histP = zeros(d,1);
        for k = 0:n
            histP = histP + b(n-k+1) * fval(:,k+1);
        end
        yp = y0 + facP * histP;
        % Assemble coefficients C(0),...,C(n+1) for
        % step n+1
        C = zeros(1,n+2);
        if n == 0
            % First step: only startup interval [t0,t1]
            C(1) = J0all(1);   % coefficient of f_0
            C(2) = J1all(1);   % coefficient of f_1
        else
            % Startup interval [t0,t1] with lag j = n
            C(1) = C(1) + J0all(n+1);
            C(2) = C(2) + J1all(n+1);
            % Intervals [t_k,t_{k+1}], k = 1,...,n
            for k = 1:n
                j = n-k;
                % Contribution to f_{k+1}
                C(k+2) = C(k+2) + Iplus(j+1);
                % Contribution to f_k
                C(k+1) = C(k+1) + Izero(j+1);
                % Contribution to f_{k-1}
                C(k) = C(k) + Iminus(j+1);
            end
        end
        % Split off the newest implicit term
        cnew  = C(n+2);
        histC = zeros(d,1);
        for m = 0:n
            histC = histC + C(m+1) * fval(:,m+1);
        end
        % Fixed-point iteration
        yold = yp;
        for it = 1:maxit
            fnew = fun(t(n+2),yold);
            fnew = fnew(:);
            if length(fnew) ~= d
                error(['fun returned inconsistent ' ...
                    'dimension ' 'at step %d.'], n+1);
            end
            ynew = y0 + facC * (histC + cnew * fnew);
            if norm(ynew - yold, inf) <= tol * max(1, ...
                    norm(ynew, inf))
                yold = ynew;
                break;
            end
            yold = ynew;
        end
        if it == maxit
            warning(['Step %d: fixed-point ' ...
                'iteration reached ' 'maxit.'], n+1);
        end
        y(:,n+2) = yold;
        ftmp = fun(t(n+2),yold);
        ftmp = ftmp(:);
        if length(ftmp) ~= d
            error(['fun returned inconsistent ' ...
                'dimension at ' 'step %d.'], n+1);
        end
        fval(:,n+2) = ftmp;
        itcount(n+1) = it;
    end
    if d == 1
        y = y(:);
    end
end
\end{lstlisting}

\newpage
\appsection{FO\_LE code}\label{b}

The code \texttt{LIL}, or \texttt{LIL\_nc}, \texttt{FO\_LE} and the extended function modeling the system must to be placed in the same folder.

For faster execution, the solver \texttt{LIL} (yellow line 96) should be replaced with the fast integrator \texttt{LIL\_nc} \cite{lilnc}.

\begin{lstlisting}
function [t,LE,tv,LEv] = FO_LE(ne,ext_fcn,...
    t_start,h_norm,t_end,x_start,h,q,out,tol,...
    maxit,alphatol)
%
% Program to compute the spectrum of Lyapunov
% exponents for non-commensurate and commensurate
% fractional-order systems using the Benettin
% algorithm and LIL_nc solver
%
% LIL_nc repository:
%
% https://www.mathworks.com/matlabcentral/fileexchange/183478-lil_nc
%
% The orthonormalization step is performed by QR
% decomposition.
%
% Marius-F. Danca 2026
%
% INPUT
%   ne        - system dimension
%   ext_fcn   - function containing the extended
%               system
%   t_start   - initial time
%   h_norm    - renormalization step
%   t_end     - final time
%   x_start   - initial condition of the original
%               system
%   h         - integration step
%   q         - fractional-order vector
%   out       - printing step of LE values (optional,
%               default 0)
%   tol       - fixed-point stopping tolerance
%               (optional, default 1e-12)
%   maxit     - max fixed-point iterations per step
%               (optional, default 100)
%   alphatol  - tolerance for grouping equal orders
%               (optional, default 1e-14)
%
% OUTPUT
%   t    - final time
%   LE   - final Lyapunov exponents
%   tv   - time values at renormalization instants
%   LEv  - Lyapunov exponent values at
%          renormalization instants
%
% Example:
%
%   [t,LE] = FO_LE(3,@RF_ext,0,0.2,1500,...
%   [0.1;0.1;0.1],0.01,[0.999;0.999;0.999],500);
%
%   RF_ext is the extended function defining
%   the Rabinovich-Fabrikant system [1]
%
%   Cite:
%   [1] Marius-F. Danca, Improved Matlab code for
%   Lyapunov exponents of fractional order,
%   submitted, 2026
%

    if nargin < 9 || isempty(out)
        out = 0;
    end
    if nargin < 10 || isempty(tol)
        tol = 1e-12;
    end
    if nargin < 11 || isempty(maxit)
        maxit = 100;
    end
    if nargin < 12 || isempty(alphatol)
        alphatol = 1e-14;
    end
    x_start = x_start(:);
    q = q(:);
    % Extended order vector
    q = repmat(q, ne+1, 1);
    % Memory allocation
    x  = zeros(ne*(ne+1),1);
    x0 = x;
    c  = zeros(ne,1);
    zn = c;
    n_it = round((t_end-t_start)/h_norm);
    tv  = zeros(n_it,1);
    LEv = zeros(n_it,ne);
    % Initial values for extended system
    x(1:ne) = x_start;
    i = 1;
    while i <= ne
        x((ne+1)*i) = 1.0;
        i = i + 1;
    end
    t = t_start;
    LE = zeros(ne,1);
    % Main loop
    it = 1;
    while it <= n_it
        % Integrate extended system on [t,t+h_norm]
        [T,Y] = LIL_nc(ext_fcn,[t,t+h_norm],x,q,h,tol,maxit,alphatol);
        t = T(end);
        x = Y(:,end);
        % Rearrange variational block
        i = 1;
        while i <= ne
            j = 1;
            while j <= ne
                x0(ne*i+j) = x(ne*j+i);
                j = j + 1;
            end
            i = i + 1;
        end
        % Build perturbation matrix
        A = zeros(ne,ne);
        i = 1;
        while i <= ne
            j = 1;
            while j <= ne
                A(i,j) = x0(ne*i+j);
                j = j + 1;
            end
            i = i + 1;
        end
        % QR orthonormalization
        [Q,R] = qr(A,0);
        d = sign(diag(R));
        d(d == 0) = 1;
        D = diag(d);
        Q = Q*D;
        R = D*R;
        zn = abs(diag(R));
        tiny = 1e-300;
        i = 1;
        while i <= ne
            if zn(i) < tiny
                zn(i) = tiny;
            end
            i = i + 1;
        end
        % Update cumulative sums and LE values
        j = 1;
        while j <= ne
            c(j) = c(j) + log(zn(j));
            j = j + 1;
        end
        LE = c/(t - t_start);
        % Store history
        tv(it)    = t;
        LEv(it,:) = LE(:).';
        % Put orthonormalized basis back into x0
        i = 1;
        while i <= ne
            j = 1;
            while j <= ne
                x0(ne*i+j) = Q(i,j);
                j = j + 1;
            end
            i = i + 1;
        end
        % Replace variational block in x
        i = 1;
        while i <= ne
            j = 1;
            while j <= ne
                x(ne*j+i) = x0(ne*i+j);
                j = j + 1;
            end
            i = i + 1;
        end
        if out > 0
            if mod(it,out) == 0
                fprintf('%10.4f ',t);
                fprintf('%12.8f ',LE);
                fprintf('\n');
            end
        end
        it = it + 1;
    end
    % Plot only at the end
    figure;
    plot(tv,LEv,'.');
    xlabel('t','fontsize',12);
    ylabel('LEs','fontsize',12);
    set(gca,'fontsize',12);
    box on;
    line([t_start t_end],[0 0],'color','k');
    axis tight
end
\end{lstlisting}

\newpage

\appsection{Function RF\_ext.m}\label{c}
\begin{lstlisting}[belowskip=0pt]%[    basicstyle=\footnotesize, %or \small or \footnotesize etc.]

function f = RF_ext(~,x)
% Extended RF system for Lyapunov exponents.
%
% State variables:
%   x(1), x(2), x(3)
%
% Variational variables:
%   x(4:12) = variational matrix entries, packed columnwise
%
% Parameters:
a = -1.0;
b = -0.1;

% Original state
x1 = x(1);
x2 = x(2);
x3 = x(3);

% Original system
f = zeros(12,1);

f(1) = x2*(x3 - 1 + x1*x1) + a*x1;
f(2) = x1*(3*x3 + 1 - x1*x1) + a*x2;
f(3) = -2*x3*(b + x1*x2);

% Jacobian of the original system
J = [2*x1*x2 + a,        x1*x1 + x3 - 1,   x2;
     -3*x1*x1 + 3*x3 + 1, a,               3*x1;
     -2*x2*x3,           -2*x1*x3,        -2*(b + x1*x2)];

% Variational matrix Phi, stored columnwise in x(4:12)
Phi = reshape(x(4:12),3,3);

% Variational equations
dPhi = J*Phi;

% Pack back columnwise
f(4:12) = dPhi(:);
end

\end{lstlisting}

\newpage

\begin{algorithm}[H]
\begin{center}
\caption{Benettin algorithm for LEs }
\label{alg1} 
\begin{minipage}{\linewidth}
\hrule\vspace{1mm}
\textbf{Input:}
    number of equations \(n_e\),
initial condition
\(x_{\mathrm{start}}\),
time interval \([0,T]\),
renormalization step \(h_{\mathrm{norm}}\), integration step \(h\),
and fractional-order vector \(\alpha\).

\begin{enumerate}[leftmargin=*,label=\textup{(\arabic*)},itemsep=1pt,topsep=2pt]
\item Initialize the perturbation matrix with the identity \(I_{n_e}\).

\item Form the extended system consisting of the original FO system and its
variational equations.

\item For each renormalization interval:
\begin{enumerate}[leftmargin=1.5em,label=\textup{(\alph*)},itemsep=1pt,topsep=1pt]
\item integrate the extended system on \([t,t+h_{\mathrm{norm}}]\);

\item extract the perturbation matrix \(\Phi\);

\item \colorbox{yellow}{\parbox{0.78\linewidth}{orthonormalize its columns by GS;}}

\item \colorbox{yellow}{\parbox{0.78\linewidth}{compute the stretching factors \(z_k\), \(k=1,\dots,n_e\);}}

\item update the cumulative sums
\[
s_k \leftarrow s_k + \log z_k,\qquad k=1,\dots,n_e;
\]

\item update the current time \(t \leftarrow t+h_{\mathrm{norm}}\), and compute
\[
\lambda_k=\frac{s_k}{t},\qquad k=1,\dots,n_e
\]
\end{enumerate}

\item Return the time evolution and/or final values of the exponents.
\end{enumerate}

\vspace{1mm}\hrule
\end{minipage}
\end{center}
\end{algorithm}

\newpage

\begin{center}
\begin{minipage}{\textwidth}
\begin{framed}
\begin{center}
\textbf{Box 1. LIL Scheme}
\label{box:lil}
\end{center}

\medskip
\noindent\textbf{Predictor}
\begin{equation}
y_{n+1}^{(p)}
=
y_0+
\frac{h^\alpha}{\Gamma(\alpha+1)}
\sum_{k=0}^{n}
b_{n-k}^{(\alpha)}\,f_k,
\qquad
f_k=f(t_k,y_k),
\label{eq:lil_box_predictor}
\end{equation}
with coefficients
\begin{equation*}
b_j^{(\alpha)}=(j+1)^\alpha-j^\alpha,
\qquad j\ge 0.
\label{eq:lil_box_b}
\end{equation*}

\medskip
\textbf{Implicit corrector}
\begin{equation}
y_{n+1}
=
y_0+
\frac{h^\alpha}{\Gamma(\alpha)}
\left[
\omega_{n+1,n+1}^{(\alpha)}\,f(t_{n+1},y_{n+1})
+
\sum_{\ell=0}^{n}
\omega_{n+1,\ell}^{(\alpha)}\,f_\ell
\right],
\label{eq:lil_box_corrector}
\end{equation}
with coefficients
\begin{equation*}
\omega_{n+1,n+1}^{(\alpha)}
=
\frac{\alpha+4}{2\alpha(\alpha+1)(\alpha+2)},
\label{eq:lil_box_omega_diag}
\end{equation*}
and, for \(\ell=0,1,\dots,n\),
\begin{equation*}
\omega_{n+1,\ell}^{(\alpha)}
=
\frac{
(n+2-\ell)^{\alpha+2}
-3(n+1-\ell)^{\alpha+2}
+3(n-\ell)^{\alpha+2}
-(n-1-\ell)^{\alpha+2}
}{
\alpha(\alpha+1)(\alpha+2)
}.
\label{eq:lil_box_omega}
\end{equation*}
\textbf{Remark.}
For the non-commensurate case, the scheme is applied componentwise: in the
\(i\)-th equation, \(i=1,2,...,ne\), \(\alpha\), \(b_j^{(\alpha)}\), and
\(\omega_{n+1,\ell}^{(\alpha)}\) are replaced by
\(\alpha_i\), \(b_j^{(\alpha_i)}\), and
\(\omega_{n+1,\ell}^{(\alpha_i)}\), respectively.
\end{framed}
\end{minipage}
\end{center}

\newpage

\begin{table*}[ht]
\tbl{Comparison of \texttt{LIL\_nc}, \texttt{fde12\_nc2}, and classical \texttt{ABM\_nc} for the non-commensurate benchmark problem. Errors are measured in the discrete maximum norm.\label{tab:R2_nc_LIL_ABM}}
{
\centering
\small
\resizebox{\textwidth}{!}{%
\begin{tabular}{c|>{\columncolor{gray!15}}c>{\columncolor{gray!15}}c>{\columncolor{gray!15}}c|>{\columncolor{gray!30}}c>{\columncolor{gray!30}}c>{\columncolor{gray!30}}c|>{\columncolor{gray!15}}c>{\columncolor{gray!15}}c>{\columncolor{gray!15}}c}
\hline
& \multicolumn{3}{c|}{\cellcolor{gray!15}\texttt{LIL\_nc}} & \multicolumn{3}{c|}{\cellcolor{gray!30}\texttt{fde12\_nc2}} & \multicolumn{3}{c}{\cellcolor{gray!15}\texttt{ABM\_nc}} \\
\hline
$h$ & $\texttt{E}_{\mathrm{LIL}}$ & $p_{\mathrm{LIL}}$ & \texttt{CPU} \texttt{LIL\_nc} & $\texttt{E}_{\mathrm{fde12\_nc2}}$ & $p_{\mathrm{fde12\_nc2}}$ & \texttt{CPU} fde12\_nc2 & $\texttt{E}_{\mathrm{ABM}}$ & $p_{\mathrm{ABM}}$ & \texttt{CPU} \texttt{ABM\_nc} \\
\hline
$0.01$ & $1.54e-04$ & -- & 0.00385 & $1.07e-04$ & -- & 0.00293 & $5.97e-03$ & -- & 0.02913 \\
$0.005$ & $5.02e-05$ & 1.62 & 0.00732 & $4.12e-05$ & 1.38 & 0.00405 & $3.43e-03$ & 0.80 & 0.10353 \\
$0.0025$ & $1.64e-05$ & 1.61 & 0.01328 & $1.47e-05$ & 1.49 & 0.00823 & $1.97e-03$ & 0.80 & 0.39542 \\
$0.00125$ & $5.39e-06$ & 1.61 & 0.02538 & $5.06e-06$ & 1.54 & 0.01677 & $1.13e-03$ & 0.80 & 1.53893 \\
$0.000625$ & $1.77e-06$ & 1.60 & 0.05995 & $1.71e-06$ & 1.57 & 0.03252 & $6.52e-04$ & 0.80 & 6.16438 \\
\hline
\end{tabular}
}
}
\end{table*}

\newpage

\begin{figure}[H]
\centering
\includegraphics[width=\columnwidth]{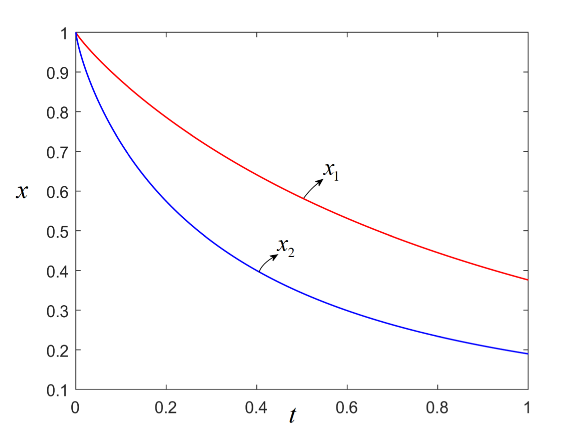}
\caption{The solutions of the problem \eqref{eq:R2_nc_test}.}
\label{fig1}
\end{figure}

\end{multicols}
\clearpage

\begin{figure}[H]
\centering
\includegraphics[width=0.95\textwidth]{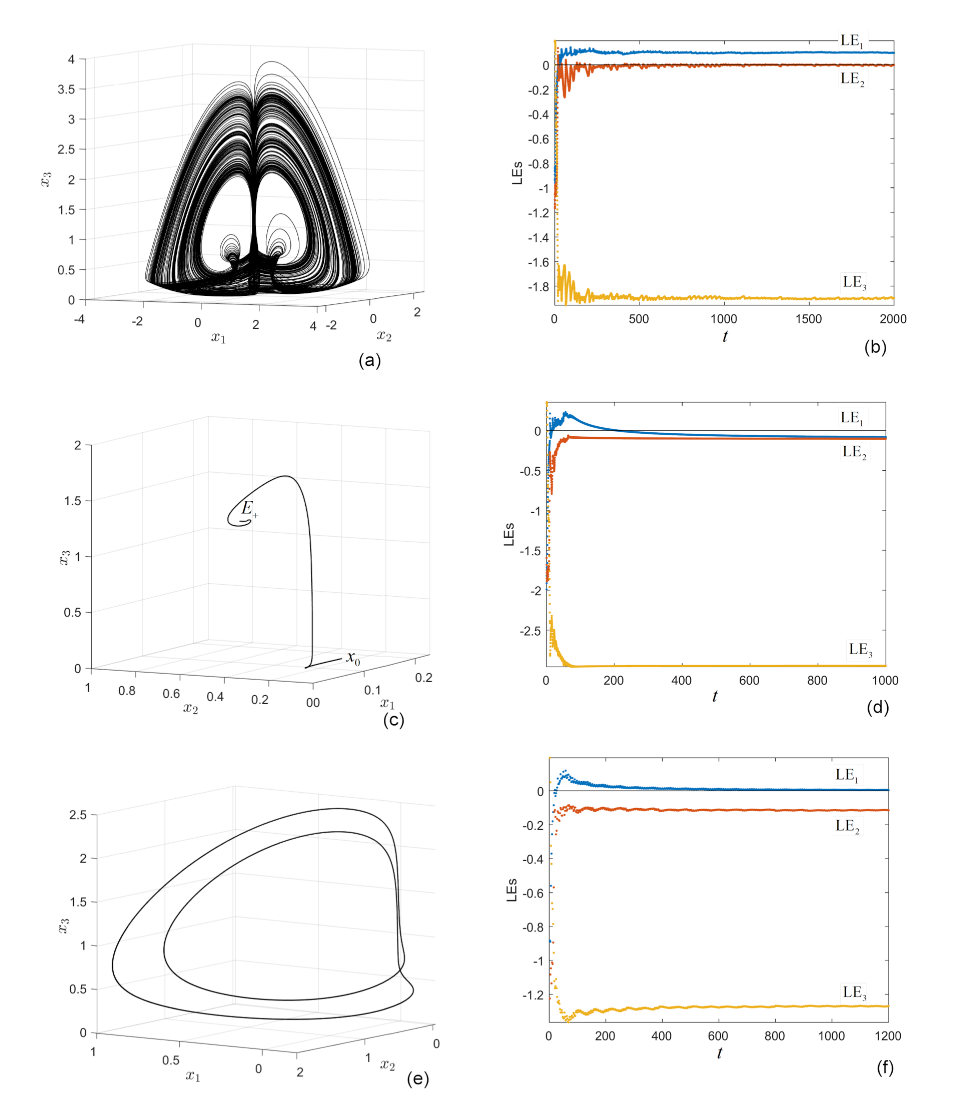}
\label{fig2}
\caption{Phase plot of the numerical solutions of obtained with thhe \texttt{LIL\_nc} integrator and time evolution of LEs, obtained with the routine \texttt{FO\_LIL}, for the RF system. (a), (b) Phase plot of the solution and time evolution of LEs for $\alpha=0.999$, respectively; (c), (d) Phase plot and time evolution of the LEs for $\alpha=(0.6,0.8,0.7)$, respectively; (e),(f) Phase plot of the solution and time evolution of LEs for $\alpha=(0.85, 0.965, 0.999)$, respectively.}
\end{figure}

\newpage{\pagestyle{empty}\cleardoublepage}
\newpage{\pagestyle{empty}\cleardoublepage}


\begin{thebibliography}{100}

	\bibitem[FO\_NC\_Lyapunov(2025)]{noncom} \url{https://www.mathworks.com/matlabcentral/fileexchange/122377-matlab-code-for-les-of-non-commensurate-fo}

	\bibitem[LIL(2026)]{lil} \url{https://www.mathworks.com/matlabcentral/fileexchange/183396-lil}

	\bibitem[LILnc(2026)]{lilnc} \url{https://www.mathworks.com/matlabcentral/fileexchange/183478-lil_nc}

	\bibitem[FO\_Lyapunov(2025)]{com} \url{https://www.mathworks.com/matlabcentral/fileexchange/my-file-exchange?s_tid=gn_mlc_fx_myfx}


	\bibitem[Baker \& Gollub(1990)]{baker} Baker, G.L. \& Gollub, P.G.[1990)] \emph{Chaotic Dynamics - An Introduction Chaotic Dynamics- An Introduction}, Cambridge University Press, New York.

	\bibitem[Benettin et al.(1980)]{bene2} Benettin, G., Galgani, L., Giorgilli, A. \& Strelcyn, J.-M. [1980].``{L}yapunov characteristic exponents for smooth dynamical systems and for hamiltonian systems. {A} method for computing all of them. {P}art {II}: Numerical application,'' {\em Meccanica}, \textbf{15}(1), 21--30.

	\bibitem[Caputo(1967)]{cap} Caputo, M. [1967] ``Linear models of dissipation whose {Q} is almost frequency independent-{II},''{\em Geophys. J. Roy. Astr. S.}, \textbf{13}(5), 529--539.

	\bibitem[Christiansen \& Rugh(1997)]{christ} Christiansen, F . \& Rugh, H. H. [1997] ``Computing Lyapunov spectra with continuous Gram - Schmidt orthonormalization,'', \emph{Nonlinearity}, \textbf{10}(5), 1063.

	\bibitem[Danca and Fe\v{c} an(2024)] {md1} Danca, M.-F. \& Fe\v{c}kan, M. [2024] ``Memory principle of the MATLAB code for Lyapunov exponents of fractional-order systems'', \textit{International Journal of Bifurcation and Chaos}, 2450156.



	\bibitem[Danca(2026)]{dd5} Danca, M.-F. [2026] ``A High Order Method for Caputo Fractional Differential Equations'', submitted.

	\bibitem[Danca \& Kuznetsov(2018)]{dd2} Danca, M.-F. \& Kuznetsov, N. [2018] ''Matlab code for Lyapunov exponents of fractional-order systems'', \textit{International Journal of Bifurcation and Chaos}, 28(5), 1850067.

	\bibitem[Danca(2021)]{dd3} Danca, M.-F. [2021] ''Matlab code for Lyapunov exponents of fractional-order systems, Part II: The non-commensurate case,''
\textit{International Journal of Bifurcation and Chaos}, 31(12), 2150187, 2021.


	\bibitem[Danca et al.(2019)]{dd4} Danca, M.-F., Bourke, P. \& Kuznetsov, N. [2019] ``Graphical structure of attraction basins of hidden attractors: the Rabinovich--Fabrikant system'', \textit{International Journal of Bifurcation and Chaos}, 29(1), 1930001.

	\bibitem[Diethlem \& Ford(2004)]{diex} Diethelm, K. \& Ford, N. J. [2004] ``Multi-order fractional differential equations and their numerical solution'', \emph{Applied Mathematics and Computation, }154, 3, 621-640

	\bibitem[Diethelm \& Ford(2002)]{existenta} Diethelm, K. \& Ford, N.J. [2002] ``Analysis of fractional differential equations,'' {\em J. Math. Anal. Appl.}, \textbf{265}(2),229--248.

	\bibitem[Diethelm et al.(2002)]{kai} Diethelm, K., Ford, N.J. \& Freed, A.D. [2002] ``A predictor-corrector approach for the numerical solution of fractional differential equations,''
{\em Nonlinear Dynam.}, \textbf{29}(1), 3--22.

	\bibitem[Eckmann \& D.~Ruelle(1985)]{lili} Eckmann, J.-P.\& Ruelle D.[1985] ``Ergodic theory of chaos and strange attractors,'' {\em Rev. Mod. Phys.}, \textbf{57}(3), 617--656.

	\bibitem[Garrappa(2018)]{gara} Garrappa, R. [2018] ``Numerical Solution of Fractional Differential Equations: A Survey and a Software Tutorial'', \emph{Mathematics }, 6(2), 16.

	\bibitem[Gorenflo \& Mainardi(1997)]{gore} Gorenflo R. \& Mainardi F.[1997] \emph{Fractional Calculus}. In: Carpinteri A., Mainardi F. (eds) Fractals and Fractional Calculus in Continuum Mechanics. International Centre for Mechanical Sciences (Courses and Lectures), vol 378. Springer, Vienna.

	\bibitem[Kilbas \& Trujillo(2001)]{kill} Kilbas, A.A. \& Trujillo, J.J.[2001]`` Differential equations of fractional order. Methods, results and problems, I,'' \emph{Appl. Anal.}, 78(1-2): 153-192.

	\bibitem[Oldham \& Spanier(1974)]{old} Oldham, K.B. Spanier, J. [1974]{\em The Fractional Calculus: Theory and Applications of Differentiation and Integration to Arbitrary Order}.
Academic Press, New York.

	\bibitem[Podlubny(2002)]{pod} Podlubny, I. [2002] ``Geometric and physical interpretation of fractional integration and fractional differentiation,'' {\em Frac. Calc. Appl. Anal.}, \textbf{5}(4):367--386.

	\bibitem[Podlubny(1999)]{pod2} Podlubny, I.[1999] \emph{Fractional Differential Equations: An Introduction to Fractional Derivatives, Fractional Differential Equations, to Methods of their Solution and some of their Applications,} Academic Press, Dan Diego.

	\bibitem[Sarra \& Meador(2011)]{sara} Sarra, S.A. \& Meador C. [2011]. ``On the numerical solution of chaotic dynamical systems using extend precision floating point arithmetic and very high order numerical methods,'' {\em Nonlinear Analysis: Modelling and Control}, \textbf{16}(3), 340--352.

	\bibitem[Shimada \& Nagashima(1979)]{shima} Shimada, I. \& Nagashima, T. [1979], ``A Numerical Approach to Ergodic Problem of Dissipative Dynamical Systems,'' \emph{Prog. Theor. Phys.} \textbf{61}(6) 1605--1616.

	\bibitem[Tavazoei \& Haeri(2009)] {tava} Tavazoei, M.S. \& Haeri, M. [2009] ``A proof for non existence of periodic solutions in time invariant fractional order systems,'' \emph{Automatica, } \textbf{45}(8), 1886--1890.

	\bibitem[Wang et al.(2012)]{sara2} Wang, P., Li, J., \& Li, Q [2012] ``Computational uncertainty and the application of a high-performance multiple precision scheme to obtaining the correct reference solution of {L}orenz equations,'' {\em Numerical Algorithms}, \textbf{59}(1), 147--159.


\end{thebibliography}
\end{document}